\documentclass[11pt]{article}

\usepackage[T1]{fontenc}
\usepackage[utf8]{inputenc}
\usepackage{lmodern}
\usepackage{microtype}
\usepackage{amsmath,amssymb,amsthm}
\usepackage[a4paper,margin=2.0cm]{geometry}
\usepackage[hidelinks]{hyperref}

\newtheorem{theorem}{Theorem}
\newtheorem{lemma}{Lemma}
\newtheorem{definition}{Definition}

\newtheorem{example}{Example}

\newcommand{\A}{\mathcal A}
\newcommand{\CT}{C_{\mathrm T}}
\newcommand{\CTsimp}{C_{\mathrm T}^{(1)}}
\newcommand{\CTsimpenv}{C_{\mathrm T}^{(1),\leq}}
\newcommand{\ellmax}{\ell_{\max}}

\title{Non-Simple T-Prescriptions Yield T-Complexity Gains Infinitely Often}
\author{Thomas Sch\"urmann\\D\"usseldorf, Germany}
\date{}

\begin{document}
\maketitle

\begin{abstract}
	Titchener et\,al.~\cite{TitchenerEtAl2005} asked whether non-simple
	T-prescriptions can achieve larger T-complexity than simple T-prescriptions
	for infinitely many finite values of the maximum codeword length. We answer
	this question affirmatively for every fixed alphabet of size at least two.
	In fact, we prove a stronger upper-envelope statement: for infinitely many
	lengths $N$, there exists a valid non-simple T-prescription of exact maximum
	codeword length $N$ whose T-complexity exceeds the best value attainable by
	any simple prescription with maximum codeword length at most $N$. Hence the
	unrestricted exact-length maximum over prescriptions is strictly larger than
	the corresponding simple exact-length maximum for infinitely many $N$.
	The proof is elementary. Once a copy pattern has been used, it cannot be
	selected again. Thus simple prescriptions with more and more steps must use
	ever more distinct words, which forces the minimal simple length thresholds
	to have infinitely many strict upward jumps. At each such jump, changing the
	last copy factor of a minimal simple prescription from $1$ to $2$ yields a
	valid non-simple prescription, thereby increasing T-complexity by
	$\log_2 3-1$ while staying below the next simple threshold.
\end{abstract}


\section{Introduction}

T-complexity is a computable deterministic complexity measure for finite words based on the recursive construction of T-codes. It belongs to the broader family of finite-string complexity measures motivated by Lempel-Ziv production complexity and by algorithmic information theory, while using a different self-learning production mechanism based on recursive catenation of previously generated codewords \cite{LempelZiv1976,Solomonoff1964,Kolmogorov1965,Chaitin1966,TitchenerEtAl2005}. The T-code construction was introduced and developed in the T-code literature, including early papers by Titchener, the work of Nicolescu--Titchener, and G\"unther's thesis \cite{Titchener1984,Titchener1996,NicolescuTitchener1998,Guenther1998}.

In this note we work with prescription-level T-complexity and with the upper-complexity problem over prescriptions of a given maximum codeword length, as in the finite examples of Section~11 of Titchener et al.~\cite{TitchenerEtAl2005}. Thus the quantities below optimize over valid prescriptions rather than directly over individual output words.

Starting with an alphabet, a T-prescription repeatedly selects a current codeword, called the copy pattern, and performs a T-augmentation at that pattern. If the selected copy patterns are $p_1,p_2,\ldots,p_m$ and the corresponding copy factors are $k_1,k_2,\ldots,k_m$, then the T-complexity of the prescription is
\begin{equation}
\CT(P)=\sum_{i=1}^{m}\log_2(k_i+1).
\end{equation}
A prescription is called simple if all copy factors are equal to one; a simple prescription with $m$ steps has $\CT(P)=m$. Throughout, \emph{unrestricted} means that no restriction is imposed on the copy factors; thus the unrestricted class contains both simple and non-simple prescriptions.

The length parameter used by Titchener et al. is the maximum codeword length of the final code generated by the prescription \cite{TitchenerEtAl2005}. We denote it by $\ellmax(P)$. For $P=((p_1,k_1),\ldots,(p_m,k_m))$, this length is
\begin{equation}
\ellmax(P)=1+\sum_{i=1}^{m}k_i|p_i|.
\end{equation}
In the simple case, this reduces to
\begin{equation}
\ellmax(P)=1+\sum_{i=1}^{m}|p_i|.
\end{equation}
These two elementary formulas already display the finite-length issue. A simple step adds one unit of T-complexity and increases the maximum codeword length by the length of the selected copy pattern. If the last simple step selects a pattern $p$, then replacing its copy factor $1$ by $2$ adds one further copy of $p$. The length increases by $|p|$, while the last complexity contribution changes from $\log_2 2=1$ to $\log_2 3$. The local gain is therefore $\log_2 3-1>0$. The question is whether this gain can be realized at infinitely many lengths before any additional simple step can fit.

Titchener et al. gave finite examples where non-simple prescriptions have strictly larger T-complexity than all simple prescriptions and explicitly asked whether this happens infinitely often \cite[Sec.~11]{TitchenerEtAl2005}. Subsequent work studied decomposition algorithms and maximum T-complexity questions from complementary algorithmic and asymptotic perspectives \cite{YangSpeidel2005,Speidel2008,ClarkTeutsch2015}. The present note gives an elementary affirmative answer to the finite-length question.

\section{Problem and Main Theorem}

Fix a finite alphabet $\A$ with $|\A|=r\geq 2$. Throughout the paper the alphabet $\A$ is fixed and the dependence on $r$ is suppressed.

Following the notation and terminology of Titchener et\,al.~\cite{TitchenerEtAl2005},
the relevant length parameter is the maximum codeword length \(\ellmax(P)\) of a
T-prescription \(P\). For an integer \(N\ge 1\), define the unrestricted
exact-length maximum by
\begin{equation}
	\CT(N)=\max\{\CT(P): P\text{ is a valid T-prescription and }\ellmax(P)=N\}.
\end{equation}
Here \emph{unrestricted} retains the convention from the introduction: all copy factors $k_i\geq 1$ are allowed. The adjective \emph{exact-length} means that the maximization is restricted to prescriptions whose final maximum codeword length is exactly $N$, not merely at most $N$.
Define the corresponding simple exact-length maximum by
\begin{equation}
\CTsimp(N)=\max\{\CT(P): P\text{ is simple and }\ellmax(P)=N\}.
\end{equation}
Here the superscript $(1)$ indicates that all copy factors are equal to one. If one of the exact-length classes is empty, its maximum is understood as $-\infty$.

We shall prove a slightly stronger statement by comparing with the simple upper envelope
\begin{equation}
\CTsimpenv(N)=\max\{\CT(P): P\text{ is simple and }\ellmax(P)\leq N\}.
\end{equation}
This monotone envelope is an additional comparison function; it strengthens the exact-length comparison used in \cite{TitchenerEtAl2005} rather than replacing it. Since
\begin{equation}
\CTsimp(N)\leq \CTsimpenv(N),
\end{equation}
a strict inequality against $\CTsimpenv(N)$ implies the strict exact-length inequality against $\CTsimp(N)$.

\begin{theorem}[Infinitely many strict non-simple gains]
\label{thm:main}
For every alphabet size $r\geq 2$, there are infinitely many lengths $N$ such that
\begin{equation}
\CT(N)>\CTsimpenv(N).
\end{equation}
In particular, there are infinitely many lengths $N$ such that
\begin{equation}
\CT(N)>\CTsimp(N).
\end{equation}
Thus non-simple T-prescriptions have strictly larger T-complexity than all simple T-prescriptions at infinitely many finite lengths.
\end{theorem}

The proof yields a brute-force exhaustive enumeration of the resulting non-simple prescriptions, but no efficiency bound or closed form for the jump indices.
For each fixed number $m$ of augmentation steps (the arity, as defined below), the set of simple prescriptions with $m$ steps is finite, so the minimal simple length threshold and a minimizing prescription can be found by exhaustive enumeration.
Computing the thresholds \(L_{m-1},L_m,L_{m+1}\) successively
allows one to test the condition \(d_{m+1}>d_m\).
At each detected jump, the required non-simple prescription is obtained by changing one copy factor from $1$ to $2$. The argument does not provide a closed-form infinite sequence of such arities or resulting lengths. No complexity bound for this enumeration, and no polynomial-time or otherwise efficient method for locating the infinitely many jump indices, is asserted.

\section{Definitions and Well-Definedness}

The alphabet $\A$ fixed above remains in force, with
\begin{equation}
|\A|=r\geq 2.
\end{equation}
All code sets considered below are finite prefix-free subsets of $\A^+$. If $p$ is a word, then $p^j$ denotes the $j$-fold concatenation of $p$, with $p^0$ the empty word.

\begin{definition}[T-augmentation]
Let $C\subset \A^+$ be a finite prefix code, let $p\in C$, and let $k\geq 1$. The T-augmentation of $C$ at the copy pattern $p$ with copy factor $k$ is
\begin{equation}
C_{(p)}^{(k)}=\bigl\{p^j s:s\in C\setminus\{p\},\ 0\leq j\leq k\bigr\}\cup\bigl\{p^{k+1}\bigr\}.
\end{equation}
This is the standard T-augmentation operation used in the T-code literature \cite{TitchenerEtAl2005,NicolescuTitchener1998,Guenther1998}.
\end{definition}

\begin{lemma}[Prefix-freeness of T-augmentation]
\label{lem:prefixfree}
If $C$ is prefix-free, then $C_{(p)}^{(k)}$ is prefix-free.
\end{lemma}

\begin{proof}
Let $x$ and $y$ be two distinct words in $C_{(p)}^{(k)}$, and suppose for contradiction that one is a prefix of the other. If both words have the form $p^j s$ and $p^h t$ with $s,t\in C\setminus\{p\}$, then the case $j=h$ would make $s$ and $t$ comparable, contradicting prefix-freeness of $C$ unless $x=y$. If, say, $j<h$, then after deleting the common prefix $p^j$, the words $s$ and $p^{h-j}t$ would be comparable. Since $h-j\geq 1$, the word $p^{h-j}t$ begins with $p$. Hence $s$ and $p$ are comparable: either $s$ is a prefix of $p$ or $p$ is a prefix of $s$. This contradicts prefix-freeness of $C$, because $s\neq p$. The same argument applies symmetrically when $h<j$.

It remains only to compare a word $p^j s$ with the terminal word $p^{k+1}$. If $p^j s$ were a prefix of $p^{k+1}$, then, after deleting $p^j$, the word $s$ would be comparable with $p$, again impossible. If $p^{k+1}$ were a prefix of $p^j s$, then $j\leq k$ and deletion of $p^j$ would imply that $p$ is a prefix of $s$, also impossible. Hence no two distinct words in $C_{(p)}^{(k)}$ are comparable.
\end{proof}

\begin{definition}[T-prescription]
Starting from
\begin{equation}
C_0=\A,
\end{equation}
a finite sequence
\begin{equation}
P=((p_1,k_1),\ldots,(p_m,k_m))
\end{equation}
is a valid T-prescription if $p_i\in C_{i-1}$ for each $i$, where $C_i$ is obtained from $C_{i-1}$ by T-augmentation at $p_i$ with copy factor $k_i$. The integer $m$ is the arity of $P$. The empty sequence is allowed; it has arity $0$, T-complexity $0$, and maximum codeword length $1$. The prescription is simple if
\begin{equation}
k_1=\cdots=k_m=1.
\end{equation}
Its T-complexity is
\begin{equation}
\CT(P)=\sum_{i=1}^{m}\log_2(k_i+1),
\end{equation}
and $\ellmax(P)$ denotes the length of the maximal-length codewords in the final code.
\end{definition}

These are the standard prefix-free T-code conventions used in the cited literature, with copy patterns written directly as words over the alphabet \cite{TitchenerEtAl2005,NicolescuTitchener1998}.

The finite optimization problems introduced above are genuine maxima. Indeed, if $\ellmax(P)\leq N$, then the length formula proved in Lemma~\ref{lem:length} implies $m\leq N-1$ and $1\leq k_i\leq N-1$ for every T-augmentation step. At each T-augmentation step the current code is finite, so only finitely many copy patterns are available. Hence, for fixed $N$, only finitely many valid prescriptions can satisfy $\ellmax(P)\leq N$. The exact-length maxima are therefore attained whenever the corresponding class is nonempty, and $\CTsimpenv(N)$ is attained for every $N\geq 1$ because the empty simple prescription is admissible.

\begin{example}[The first binary instance]
\label{ex:binary}
Let $\A=\{0,1\}$. Consider the simple prescription
\begin{equation}
P_2=((0,1),(1,1)).
\end{equation}
After the first step the code is $\{00,01,1\}$, so the second copy pattern $1$ is available, and, as proved below in Lemma~\ref{lem:length}, $\ellmax(P_2)=3$ and $\CT(P_2)=2$. Replacing only the last factor by $2$ gives
\begin{equation}
\widehat P_2=((0,1),(1,2)),
\end{equation}
with
\begin{equation}
\ellmax(\widehat P_2)=4,
\qquad
\CT(\widehat P_2)=1+\log_2 3=\log_2 6>2.
\end{equation}
No simple prescription with three steps has maximum codeword length at most $4$: by Lemma~\ref{lem:norecycling}, as proved below, it would have to select three distinct nonempty binary words, whose total length is at least $1+1+2=4$, and hence Lemma~\ref{lem:length}, also proved below, would give maximum codeword length at least $5$. Thus $N=4$ is the first occurrence of the same mechanism that will be used infinitely often; this is the initial finite phenomenon displayed in Titchener et al.~\cite[Sec.~11]{TitchenerEtAl2005}.
\end{example}

\section{Proof of the Theorem}
\label{sec:proof}

We first record two elementary structural facts. They are standard consequences of the recursive T-code construction, but short proofs are included to make the argument self-contained.

\begin{lemma}[Maximum codeword length]
\label{lem:length}
Let
\begin{equation}
P=((p_1,k_1),\ldots,(p_m,k_m))
\end{equation}
be a valid T-prescription. Then
\begin{equation}
\ellmax(P)=1+\sum_{i=1}^{m}k_i|p_i|.
\end{equation}
In particular, if $P$ is simple, then
\begin{equation}
\ellmax(P)=1+\sum_{i=1}^{m}|p_i|
\end{equation}
and
\begin{equation}
\CT(P)=m.
\end{equation}
\end{lemma}

\begin{proof}
The initial code is the alphabet, hence its maximal-length codewords have length one. Consider one augmentation at a current copy pattern $p$ with copy factor $k$, and suppose that the old maximum codeword length is $L$. Every new codeword has the form $p^j s$, with $s\neq p$ and $0\leq j\leq k$, or the form $p^{k+1}$. The first type has length at most $k|p|+L$, and the second type has length $(k+1)|p|\leq k|p|+L$, because $|p|\leq L$. Therefore no new codeword has length greater than
\begin{equation}
L+k|p|.
\end{equation}
This length is attained. If $p$ is not a maximal-length old codeword, take a maximal-length old codeword $s\neq p$ and obtain $p^k s$. If $p$ is itself a maximal-length old codeword, then $p^{k+1}$ has length $L+k|p|$. Thus every step increases the maximum codeword length by exactly $k|p|$. Summing the increments proves the formula. The simple case follows from $k_i=1$.
\end{proof}

\begin{lemma}[No recycling of copy patterns]
\label{lem:norecycling}
In a valid T-prescription, a copy pattern can never be selected twice.
\end{lemma}

\begin{proof}
Suppose that $p$ is selected at some step with copy factor $k$. After this augmentation, $p$ is removed from the code and becomes a proper prefix of the new terminal codeword $p^{k+1}$.
We claim that this property persists: after every later augmentation, $p$ is still a proper prefix of at least one current codeword. Indeed, assume that $p$ is a proper prefix of some current codeword $w$. If the next copy pattern is not $w$, then $w$ remains in the code, because every non-selected old codeword occurs with prefix power $0$. If the next copy pattern is $w$, then $w$ is removed, but the terminal word $w^{h+1}$ is inserted for the corresponding copy factor $h$, and $p$ is still a proper prefix of $w^{h+1}$. This proves the invariant by induction.
Since all intermediate code sets are prefix-free by Lemma~\ref{lem:prefixfree}, a word that is a proper prefix of a current codeword cannot itself be a current codeword. Thus $p$ is unavailable for every later step and cannot be selected twice.
\end{proof}

For $m\geq 0$, let $L_m$ be the smallest possible maximum codeword length after exactly $m$ simple T-augmentations:
\begin{equation}
L_m=\min\{\ellmax(P): P\text{ is simple and has arity }m\}.
\end{equation}
Here the minimum is taken over all valid simple T-prescriptions $P$ with exactly $m$ augmentation steps, i.e. over all dynamically admissible choices of $m$ copy patterns with all copy factors equal to one.
The class is nonempty for every $m$: after every finite sequence of T-augmentations the current code is finite and nonempty, so at least one leaf/codeword is available for another simple step. For fixed $m$ there are only finitely many simple prescriptions of arity $m$; this follows by induction, because from each finite intermediate code only finitely many next copy patterns can be chosen. Hence the minimum exists. Thus
\begin{equation}
L_0=1.
\end{equation}
Put
\begin{equation}
d_m=L_m-L_{m-1}\qquad (m\geq 1).
\end{equation}
Indeed, let $P$ be a minimizing simple prescription of arity $m$, and let $P^-$ be its prefix of arity $m-1$. By the definition of $L_{m-1}$, one has $\ellmax(P^-)\geq L_{m-1}$. By Lemma~\ref{lem:length}, the last simple step increases $\ellmax$ by the positive length of its copy pattern, hence by at least $1$. Therefore $L_m=\ellmax(P)\geq L_{m-1}+1$, and so
\begin{equation}
d_m\geq 1\qquad (m\geq 1).
\end{equation}
The number $d_m$ is the increase in the minimal achievable maximum codeword length when the arity is raised from $m-1$ to $m$ within the class of simple prescriptions.

\begin{lemma}[Threshold increments have infinitely many upward jumps]
\label{lem:jumps}
There are infinitely many indices $m$ such that
\begin{equation}
d_{m+1}>d_m.
\end{equation}
\end{lemma}

\begin{proof}
The proof rests on the following elementary counting principle. By Lemma~\ref{lem:norecycling}, a simple prescription with $m$ steps must select $m$ distinct copy patterns. By Lemma~\ref{lem:length}, these selected word lengths are exactly the $m$ positive increments whose sum, plus the initial length $1$, gives the final maximum codeword length. Thus, if a simple prescription has many steps, it must use many distinct nonempty words; but only finitely many nonempty words are short.

Let $D_{r,m}$ denote the smallest possible total length of $m$ distinct nonempty words over an alphabet of size $r$. For any simple prescription $P$ of arity $m$, Lemma~\ref{lem:norecycling} gives distinct copy patterns $p_1,\ldots,p_m$, and Lemma~\ref{lem:length} gives
\begin{equation}
\ellmax(P)=1+\sum_{i=1}^{m}|p_i|\geq 1+D_{r,m}.
\end{equation}
Taking the minimum over all simple prescriptions of arity $m$ yields
\begin{equation}
L_m\geq 1+D_{r,m}.
\end{equation}
This is only a lower bound: $D_{r,m}$ ignores the additional dynamical constraint that the selected words must be available at the required times during one valid T-augmentation process. The lower bound is sufficient.

It remains to show that $D_{r,m}/m\to\infty$. Fix an integer $H\geq 1$ and consider the levels of nonempty words by length:
\begin{equation}
\A^1,\quad \A^2,\quad \ldots,\quad \A^H,\quad \A^{H+1},\ldots .
\end{equation}
The number of nonempty words of length at most $H$ is
\begin{equation}
B_{r,H}=\sum_{j=1}^{H}r^j=\frac{r^{H+1}-r}{r-1}.
\end{equation}
Therefore, among any $m$ distinct nonempty words, at most $B_{r,H}$ can have length at most $H$. At least $\max\{0,m-B_{r,H}\}$ of them must have length at least $H+1$. Consequently, for every $H\geq 1$,
\begin{equation}
D_{r,m}\geq \max\{0,m-B_{r,H}\}(H+1).
\end{equation}
Now choose
\begin{equation}
H=H_m:=\left\lfloor \frac{1}{2}\log_r m\right\rfloor .
\end{equation}
For all sufficiently large $m$ this integer is at least $1$, and
\begin{equation}
B_{r,H_m}=\frac{r^{H_m+1}-r}{r-1}\leq \frac{r\sqrt m-r}{r-1}=O(\sqrt m).
\end{equation}
Hence $B_{r,H_m}/m\to 0$, while $H_m+1\to\infty$. For all sufficiently large $m$, the maximum in the preceding lower bound is therefore attained by $m-B_{r,H_m}$, and division by $m$ gives
\begin{equation}
\frac{D_{r,m}}{m}
\geq \left(1-\frac{B_{r,H_m}}{m}\right)(H_m+1)
\longrightarrow \infty.
\end{equation}
Thus $D_{r,m}/m\to\infty$, and consequently $L_m/m$ is unbounded.

If the increments $d_m=L_m-L_{m-1}$ were bounded, say $d_m\leq M$ for all $m$, then
\begin{equation}
L_m=L_0+\sum_{j=1}^{m}d_j\leq 1+Mm,
\end{equation}
contradicting the unboundedness of $L_m/m$. Hence the sequence $(d_m)_{m\geq 1}$ is unbounded. Finally suppose that only finitely many strict upward jumps $d_{m+1}>d_m$ occurred. After the last such jump the sequence would be non-increasing, and the finite initial segment is bounded as well; hence the whole sequence $(d_m)$ would be bounded. This contradiction proves that infinitely many indices satisfy $d_{m+1}>d_m$.
\end{proof}

\begin{proof}[Proof of Theorem~\ref{thm:main}]
Choose $m$ such that
\begin{equation}
d_{m+1}>d_m.
\end{equation}
Let $P_m$ be a simple prescription with arity $m$ and minimal maximum codeword length:
\begin{equation}
\ellmax(P_m)=L_m.
\end{equation}
Write
\begin{equation}
P_m=((p_1,1),\ldots,(p_m,1)).
\end{equation}
Let $P_m^-$ be the prefix prescription consisting of the first $m-1$ steps. Since $P_m^-$ is simple of arity $m-1$, the definition of $L_{m-1}$ gives
\begin{equation}
\ellmax(P_m^-)\geq L_{m-1}.
\end{equation}
The last simple step from $P_m^-$ to $P_m$ increases the maximum codeword length by $|p_m|$. Therefore
\begin{equation}
|p_m|=\ellmax(P_m)-\ellmax(P_m^-)
       \leq L_m-L_{m-1}=d_m.
\end{equation}
Now keep the first $m-1$ steps and replace only the last simple step $(p_m,1)$ by the non-simple step $(p_m,2)$. This gives the valid prescription
\begin{equation}
\widehat P_m=((p_1,1),\ldots,(p_{m-1},1),(p_m,2)).
\end{equation}
It is valid because $p_m$ is available after the first $m-1$ steps. By Lemma~\ref{lem:length}, the first $m-1$ increments are unchanged, while the last increment is $2|p_m|$ instead of $|p_m|$. Hence its maximum codeword length is
\begin{equation}
N_m=\ellmax(\widehat P_m)=\ellmax(P_m)+|p_m|=L_m+|p_m|.
\end{equation}
Using $|p_m|\leq d_m$ and $d_{m+1}>d_m$, we get
\begin{equation}
N_m\leq L_m+d_m<L_m+d_{m+1}=L_{m+1}.
\end{equation}
Thus no simple prescription with arity $m+1$ or larger can have maximum codeword length at most $N_m$. To see this explicitly, suppose that a simple prescription $Q$ has arity $q\geq m+1$ and $\ellmax(Q)\leq N_m$. Let $Q^{[m+1]}$ be its prefix of length $m+1$. By Lemma~\ref{lem:length}, every later simple step adds a positive word length, so
\begin{equation}
\ellmax(Q^{[m+1]})\leq \ellmax(Q)\leq N_m<L_{m+1}.
\end{equation}
This contradicts the definition of $L_{m+1}$ as the minimal maximum codeword length among all simple prescriptions of arity $m+1$.

Hence every simple prescription counted by $\CTsimpenv(N_m)$ has at most $m$ steps, and so
\begin{equation}
\CTsimpenv(N_m)\leq m.
\end{equation}
On the other hand,
\begin{equation}
\CT(\widehat P_m)=(m-1)+\log_2 3>m.
\end{equation}
Since $\ellmax(\widehat P_m)=N_m$, this gives
\begin{equation}
\CT(N_m)\geq \CT(\widehat P_m)>m\geq \CTsimpenv(N_m).
\end{equation}
There are infinitely many indices $m$ with $d_{m+1}>d_m$ by Lemma~\ref{lem:jumps}. The corresponding lengths $N_m$ are unbounded: since $d_j\geq 1$ for every $j\geq 1$,
\begin{equation}
L_m=L_0+\sum_{j=1}^{m}d_j\geq 1+m\longrightarrow\infty,
\end{equation}
and $N_m\geq L_m$. Hence infinitely many distinct lengths satisfy the claimed strict inequality. Since
\begin{equation}
\CTsimp(N)\leq \CTsimpenv(N),
\end{equation}
the exact-length inequality follows for the same lengths.
\end{proof}

\section{Conclusion}

The finite-length question raised by Titchener et\,al.~\cite[Sec.~11]{TitchenerEtAl2005}
has an affirmative answer. For every alphabet with at least two symbols, there are
infinitely many maximum codeword lengths at which the unrestricted exact-length
maximum T-complexity exceeds the simple upper envelope, and hence also the
corresponding simple exact-length maximum.

The mechanism is elementary. Once a copy pattern has been used, it cannot be
selected again; therefore simple prescriptions with many steps must use many
distinct copy patterns, while only finitely many short words are available. This
forces the minimal simple thresholds to have increments with infinitely many
strict upward jumps. At each such jump, changing only the last copy factor of a
minimal simple prescription from $1$ to $2$ stays below the next simple threshold
and increases T-complexity by $\log_2 3-1$. The constructed prescription is not
asserted to be an unrestricted maximizer; the strict inequality for the
unrestricted maximum follows because this maximum dominates every valid
prescription of the same exact length.

The result is deliberately finite-length in nature. It does not describe the density
or distribution of the exceptional lengths, determine the largest possible
non-simple gain, or provide a closed-form sequence of resulting lengths or an
efficient procedure for locating the jump indices. It complements the asymptotic
and algorithmic work on T-complexity by showing that the advantage of non-simple
copy factors is not merely an isolated short-length phenomenon.

\section*{Acknowledgements}

The author is grateful to Mark R. Titchener for his helpful comments and suggestions, which improved the presentation and helped clarify the connection with the T-code framework.

\end{document}